\documentclass{elsarticle}

\usepackage{amsmath,amsthm,amssymb,graphicx,xspace,verbatim,amsfonts,mathrsfs}

\usepackage{paralist} 
 
\theoremstyle{plain} 
\newtheorem{thm}{Theorem}
\newtheorem{lem}[thm]{Lemma} 
\newtheorem{prop}[thm]{Proposition}
\newtheorem{cor}[thm]{Corollary}

\theoremstyle{definition} 
\newtheorem*{defn}{Definition}
\newtheorem*{example}{Example}
\newtheorem*{rem}{Remark}
 
\newcommand{\MRCA}{\operatorname{MRCA}}
\newcommand{\desc}{\operatorname{desc}_\mathcal X}

\newcommand{\PP}{\mathbb P}
\newcommand{\tc}{\text{:}} 

\usepackage{hyperref}
\hypersetup{backref,  pdfpagemode=FullScreen,  urlcolor=blue, colorlinks=true, citecolor=blue, hyperindex=true}

\begin{document}

\begin{frontmatter}

\title{Determining species tree topologies from clade probabilities under the coalescent}

\author[UAF]{Elizabeth S. Allman}
\author[JHD]{James H. Degnan}
\author[UAF]{John A. Rhodes}
\address[UAF]{Department of Mathematics and Statistics, University of Alaska Fairbanks, PO Box 756660, Fairbanks, AK 99775, U.S.A}
\address[JHD]{Department of Mathematics and Statistics, University of Canterbury, Private Bag 4800, Christchurch, New Zealand}

\begin{abstract} 
One approach to estimating a species tree from a collection of gene trees is to first estimate probabilities of clades from the gene trees, and then to construct the species tree from the estimated clade probabilities. While a greedy consensus algorithm, which consecutively accepts the most probable clades compatible with previously accepted clades, can be used for this second stage, this method is known to be statistically inconsistent under the multispecies coalescent model.  This raises the question of whether it is theoretically possible to reconstruct the species tree from known probabilities of clades on gene trees.

We investigate clade probabilities arising from the multispecies coalescent model, with an eye toward identifying features of the species tree. Clades on gene trees with probability greater than 1/3 are shown to reflect clades on the species tree, while those with smaller probabilities may not. Linear invariants of clade probabilities are studied both computationally and theoretically, with certain linear invariants giving insight into the clade structure of the species tree. For species trees with generic edge lengths, these invariants can be used to identify the species tree topology. These theoretical results both confirm that clade probabilities contain full information on the species tree topology and suggest future directions of study for developing statistically consistent inference methods from clade frequencies on gene trees.
\end{abstract}


\end{frontmatter}

\section {Introduction}

A fundamental problem in evolutionary biology is to determine relative
relatedness of species, usually by seeking a rooted tree that diagrammatically depicts these relationships.  
Although phylogenetic methods of inferring relationships between genes sampled from
individuals in the different species are now highly developed,  such gene trees are not species trees.
Even in the absence of errors due to estimating gene trees from DNA sequences, gene tree topologies need not match the underlying species tree.
In recent years, various methods have been proposed for inferring species trees from genetic data \cite{degnan2009,edwards2009,KK}. 
Many of these methods
first estimate gene trees, and then 
resolve the possible conflicts among them to obtain an overall estimate of the species tree.

An important cause of gene tree conflict is the population effect of \emph{incomplete lineage sorting}, in which gene lineages  coalesce in ancestral populations earlier than the time these lineages first enter a common ancestral population.
The \emph{multispecies coalescent model} 
\cite{pamilo1988,rosenberg2002,rannala2003,degnan2005,degnan2009}
is commonly used to model this process, producing a distribution of rooted gene trees given a rooted species tree
topology and branch lengths (a measure of time and population size on
each edge of the species tree).  The multispecies coalescent provides 
a natural framework for incorporating population effects, allowing gene trees to possibly be 
discordant with the species tree (see Fig.~\ref{F:partition}), a phenomenon that is
very common in multilocus studies
\cite{rokas2003,ebersberger2007,cranston2009}.

\begin{figure}[!t]
\begin{center}
\setlength{\unitlength}{0.05 mm}%
   \begin{picture}(100,1204)(0,0)
   \put(-1000,200){\fontsize{14.23}{17.07}\selectfont   \makebox(2436.0, 0.0)[l]{$a$ \;\;\;\;\;\;\;$b$ \;\;\;\;$c$\; \;\;$d$\; \;$e$\strut}}
    \put(-100,650){\fontsize{11.23}{14.07}\selectfont \makebox(2436.0,0.0)[l]{$\MRCA(\{a,b,c\})$\strut}}
\put(-1250,0){
\includegraphics[width=.48\textwidth]{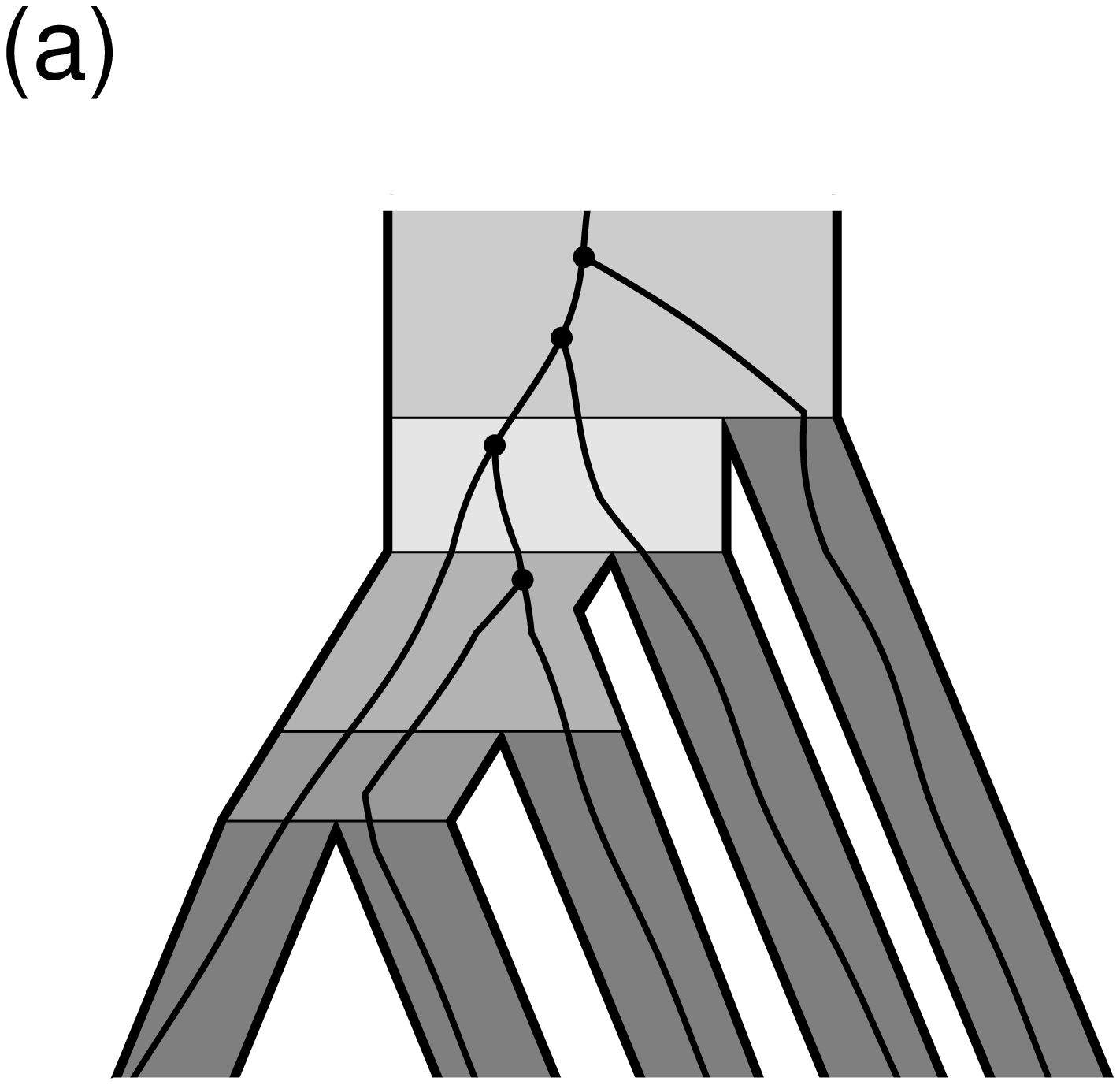}}
\end{picture}
   \begin{picture}(100,1204)(0,0)
      \put(300,200){\fontsize{14.23}{17.07}\selectfont   \makebox(2436.0, 0.0)[l]{$a$ \;\;\;\;\;\;\;$b$ \;\;\;\;$c$\; \;\;$d$\; \;$e$\strut}}
\put(50,0){
\includegraphics[width=.48\textwidth]{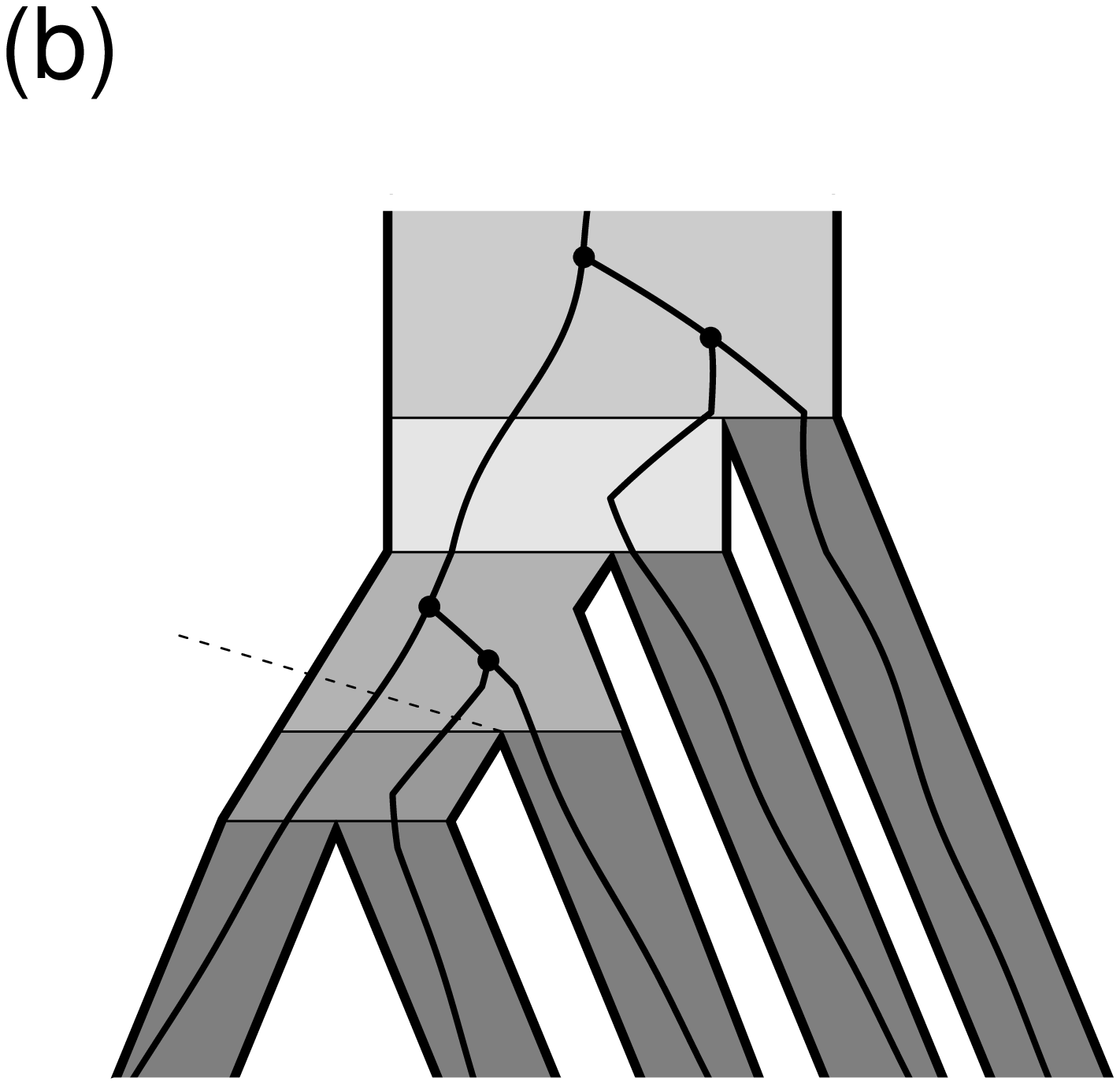}}
\end{picture}
\caption{Gene trees within a species tree.  In the multispecies
  coalescent, gene lineages sampled from species are assumed to
  coalesce (form nodes in the gene tree) no more recently than their
  most recent common ancestor (MRCA) in the species tree.  Coalescence
  of lineages in populations more ancient than their MRCA can lead to
  gene tree topologies that are discordant with the species tree
  topology. Using upper case letters for gene lineages sampled from
  their corresponding species, failure of the $A$ and $B$ lineages to
  coalesce in their MRCA population makes any of the $\binom{3}{2}$
  coalescences between $A$, $B$, and $C$ equally likely under the
  model in the MRCA population of $a$, $b$, and $c$.  (a) The gene
  tree is $((((B,C),A),D,E)$. (b) The gene tree is
  $(((B,C),A),(D,E))$.}
\label{F:partition}
\end{center}
\end{figure}

Although the distribution of gene tree topologies from the
multispecies coalescent determines the species tree \cite{adr2011},
estimating this distribution is difficult
because there are so many possible topologies: $(2n-3)!!$ when $n$ species are under study.
Thus most topologies are unlikely to be observed among a moderate
number of gene trees.  An alternative is to estimate a smaller set of
probabilities which is a function of gene tree probabilities but that
still retains enough information to identify the species tree.  Other works have considered rooted triples \cite{degnanEtAl2009,ewing2008,liuBMC2010} and unrooted gene tree topologies \cite{adr2011,larget2010}. Another
possibility, which is our focus here, is to use probabilities that a gene tree has a given
\emph{clade}, a set of leaves descended from a 
node of the gene tree
that is not ancestral to any other leaves in the gene
tree.  The probability of a clade under the
multispecies coalescent (or any model of gene tree generation) is obtained by simply adding the probabilities of all gene trees
that display the given clade \cite{degnanEtAl2009}.

The probability of a clade can be
estimated from a collection of gene trees by considering the
proportion of gene trees displaying the clade.  Since this procedure does
not take into account uncertainty in the gene trees, which are
themselves estimates from genetic data, a more sophisticated method
would quantify the uncertainty in the clades by using posterior
probabilities or bootstrap support values for clades obtained from
Bayesian or maximum likelihood analyses of the gene trees.  The
software BUCKy \cite{ane2007}, for example, takes this approach, using posterior
probabilities for clades and additionally incorporating a prior
distribution for the amount of gene tree conflict to yield a
\emph{concordance factor} for each clade.

\medskip

One of the most straightforward methods for constructing a species tree from clade probabilities is to use 
\emph{greedy consensus}, in which the clade with the highest probability (or concordance factor) 
is accepted, provided it is compatible with previously accepted clades. This process is repeated 
until a fully resolved tree is formed \cite{bryant2003}.  This procedure is implemented in  BUCKy to 
construct a \emph{concordance tree}, which is sometimes interpreted as an estimated species 
tree \cite{cranston2009}.  

To justify a greedy approach, one needs to investigate whether the most probable clades tend also to be 
clades on the species tree.
Indeed, we show in Section \ref{sec:high} that under the multispecies 
coalescent, any clade with probability greater than 1/3 must be on the species tree, suggesting that the standard majority-rule consensus (which only accepts clades occurring more 
than 50\% of the time) is very conservative in this setting.   If the greedy 
consensus approach is used for clades with probability greater than 1/3 (leaving the tree 
unresolved with respect to clades with lower probability), then this ``not-too-greedy" consensus 
approach is not misleading, in the sense that it asymptotically cannot return a false species 
tree clade as the number of loci approaches infinity.  

In contrast, previous results have shown that when greedy consensus is applied
without restrictions on clade probabilities, the returned tree can be
misleading (\emph{i.e.}, for some species trees, as the number of loci increases, 
the greedy consensus method is
increasingly likely to produce a tree that disagrees with the true
species tree) for some sets
of branch lengths \cite{degnanEtAl2009}. These ``too-greedy zones" of
edge lengths occur on $4$-taxon asymmetric species trees and on any
species tree topology with five or more leaves.  Thus, caution must be
used when probabilities of clades are less than 1/3; it is not
obvious how to determine which low-probability
clades are on the species tree,
even if clade probabilities are known exactly.  Other examples show
that the most probable $k$-clade (a clade of $k\ge 2$ elements), is
not necessarily a clade on the species tree, even if the species tree
is known to have a $k$-clade.
 
Undeterred by these negative results, we show in Sections \ref{sec:inv} and \ref{sec:id} that under the
multispecies coalescent with one lineage sampled per species, the set
of clade probabilities does identify the species tree topology for
generic branch lengths for any number of species.  The proof is based on discovering a linear combination of clade
probabilities (a linear invariant) that is equal to zero for any
branch lengths on any species tree with a given clade.
In theory, if clade probabilities are known,
it is therefore possible to identify the species tree by determining
all of its clades.  While this suggests a statistically consistent method of inferring a species tree from estimated clade probabilities, it remains a challenge to incorporate this insight into a practical method
that outperforms greedy consensus on most finite data
sets.

Finally, in Section \ref{sec:id} we extend our results,  in part, to cases where the species tree is
non-binary and where an arbitrary number of lineages
is sampled per species.

\medskip

Although we frame our questions within the framework of the multispecies coalescent, a careful reading of our arguments reveals that the essential feature of the model that  we use is that lineages are \emph{exchangeable}. If two gene lineages are present in the same population at a particular point in time on the species tree, then above that point, the model assumes that both lineages 
behave the same way. Much of this work, then, should be robust to variations on the coalescent model that preserve exchangeability.

\medskip

On a more technical note, there is a key difference in understanding clade probabilities versus many
other sets of probabilities related to gene trees or species trees: the failure of marginalization arguments. As this difference plays an important, but unspoken, part throughout this work, we highlight it here.

The problem of establishing identifiability of a species tree from unrooted gene tree probabilities that was taken up in
\cite{adr2011} is superficially similar to the clade problem of this paper.  Both unrooted gene tree
probabilities and clade probabilities can be obtained by summing
probabilities of appropriate rooted gene trees.  The sum is either
over all rooted gene trees with the same unrooted topology, or over all
rooted gene trees that have the clade in question.  

Note also
that the probability of a gene tree on a subset of the taxa can be obtained by
summing probabilities of gene trees on the full set that display the
given gene tree when restricted to the subset.  
Such marginalization
of a gene tree distribution to fewer taxa is possible for either
rooted or unrooted gene trees.  Consequently, for most arguments in \cite{adr2011} it was
sufficient to focus on small trees, with at most five taxa. Indeed, similar marginalization
arguments are standard throughout phylogenetic theory.

Unfortunately, a marginalization approach fails for studying clades when the species tree is unknown.
Given clade probabilities arising from an $n$-taxon species  tree $\sigma$ under the multispecies coalescent,
one would like to be able to determine clade probabilities arising from an induced $k$-taxon 
tree displayed on $\sigma$.
However, probabilities of clades on the $k$-taxon induced tree
cannot be obtained from a linear combination of the clade
probabilities associated with the $n$-taxon species tree in a way that
is independent of the species tree topology.  
We demonstrate this formally
is the case where $k=3$ and $n=4$ in
Appendix \ref{ap:linear}.  

This inability to marginalize clade
probabilities without knowing the species tree topology motivated looking for
an invariant that would hold for clades on trees of any size.
Although only linear invariants are needed in the proof of
identifiability, the invariants constructed for $k$-clades involve a
linear combination of $2^{k-1}$ clade probabilities.  These rather
elaborate invariants and the inability to marginalize clade
probabilities to smaller trees lead to a different flavor for the
proof of species tree identifiability from clade probabilities.

\section{Definitions}\label{sec:defs}

Let $\mathcal X$ be a finite set, whose elements we refer to as
\emph{taxa}.  A \emph{species tree on} $ \mathcal X$ means a pair
$\sigma = (\psi, \lambda)$, where $\psi$ is a rooted, topological tree
whose leaves are bijectively labelled by elements of $ \mathcal X$,
and $\lambda = (\lambda_1, \dots, \lambda_k)$ is a collection of
lengths for the internal branches of $\psi$. We refer to $\psi$ as a
\emph{species tree topology}, and always assume all internal nodes of
$\psi$ except the root have degree at least 3. If all internal nodes
except the root have degree 3 and the root has degree 2, we say that $\psi$
and $\sigma$ are \emph{binary}.

We use a modified Newick notation for species trees, as in
\cite{adr2011}, in which we do not specify the lengths of pendant
edges, since only the lengths of internal edges affect probabilities
of gene tree topologies under the multispecies coalescent.  For
example, we write $((a,b)\tc t, c)$ for a 3-taxon species tree with
one internal edge with length $t$, measured in coalescent
units.  If there is a constant effective population size,  $N$, over 
an edge of the species tree, then a length of $t$ indicates that the edge 
represents $Nt$ generations \cite{degnan2009}. For varying effective 
population size, a non-linear scaling is needed to relate coalescent units 
to generations.   Species trees are thus not assumed to be ultrametric in coalescent units.

In discussing trees, we find it convenient in various settings to use
either spatial or temporal terminology. For instance, if $(v,w)$ is a
directed edge in $\psi$ pointing away from the root, then we may say
that $v$ is \emph{above}, or an \emph{ancestor}, of $w$ and that $w$ is
\emph{below}, or a \emph{descendant}, of $v$.  Natural extensions of
these terms should be clear from context.

\smallskip

We denote taxa in $\mathcal X$ by lower case letters such
as $a,b,c,\dots$. To distinguish between taxa and sampled genes
from those taxa, we use the corresponding upper case letters
$A,B,C,\dots$ to denote the genes, with the set $\mathcal X_g$
denoting the full set of genes, one for each taxon. Similarly, a subset of taxa 
$\mathcal C \subseteq \mathcal X$ has a corresponding subset of genes 
$\mathcal C_g \subseteq \mathcal X_g$.  A sampled gene
tree from the multispecies coalescent model on $\sigma$ will thus have
leaves labelled by $\mathcal X_g$, and in general may have any
topology, regardless of the species tree topology $\psi$. More
specifically, by a \emph{gene tree} $T$ we mean a binary, rooted
topological tree with leaves bijectively labelled by $ \mathcal
X_g$.  We emphasize that for this article gene trees are topological
only, with no edge lengths specified. We require that gene trees be
binary, since under the multispecies coalescent only binary gene trees
have positive probability.

\begin{defn}
  If $\psi$ is a species tree topology on $ \mathcal X$, and $\mathcal
  A\subseteq \mathcal X$, then the \emph{most recent common ancestor
  of $\mathcal A$}, $\MRCA(\mathcal A)$, is the node of $\psi$
  that is ancestral to all elements of $\mathcal A$ and which is a descendant of
  any other node ancestral to all elements of $\mathcal A$.
\end{defn}

\begin{defn} Let $\desc(v)\subseteq \mathcal X$ denote the elements of
  $\mathcal X$ descended from a node $v$ of a species tree topology
  $\psi$ on $\mathcal X$, so that if $\mathcal A\subseteq \mathcal X$,
  then $\mathcal A\subseteq \desc(\MRCA(\mathcal A))\subseteq \mathcal
  X$. A \emph{clade} $\mathcal C$ on $\psi$ is a subset of $\mathcal
  X$ such that $\mathcal C=\desc(\MRCA(\mathcal C)).$
\end{defn}

The notions of MRCA and clade extend to gene trees in an obvious way,
replacing $\mathcal X$, $\psi$, $\mathcal C$, with $\mathcal X_g$,
$T$, $\mathcal C_g$ in the definitions.

\begin{defn} For a gene tree $T$, the set of all clades on $T$ is
  denoted $\mathcal H(T)$. Similarly, for a species tree
  $\sigma=(\psi,\lambda)$ the set of clades on $\psi$ is denoted
  $\mathcal H(\sigma)=\mathcal H(\psi)$.
\end{defn}

In discussing the relationships between a subset $\mathcal Y$ of the taxa $\mathcal X$ on a tree $\psi$, we use the terminology of a \emph{displayed tree}:
a tree obtained from the full tree by first passing to the rooted subtree spanned by $\mathcal Y$, and then
suppressing any non-root nodes of degree 2 \cite{semple2003}. As an example, the species tree in Fig.~\ref{F:partition} displays $((b,d),e)$.
The notion of displayed trees 
can be applied in the context of either species trees (with or without branch lengths) or gene trees.

\smallskip

A detailed presentation of the multispecies coalescent model is given
in \cite{adr2011}, so we omit repeating that here. 
Because we focus in this paper on the probabilities of observing
gene trees or clades on gene trees under that model, we fix the following notation.

\begin{defn}
  Under the multispecies
  coalescent model on a fixed species tree $\sigma$ on taxa $\mathcal X$, the probabilities
  of a gene tree $T$, and a clade 
  $\mathcal C_g$ on gene trees are denoted $\PP_\sigma(T)$ and
  $\PP_\sigma(\mathcal C_g)$, respectively.
\end{defn}

If more than one lineage is sampled per species, a generalization
of our results on species tree identifiability still holds.  For this extension, we
require the following definitions. 
(See Fig.~\ref{F:intra} for an example.)
\begin{defn}\label{def:extended}
  Let $\mathcal X = \{x_{1}, \dots, x_{n}\}$ be a taxon set, with
  $|\mathcal X| = n$.  Let $\delta = (\delta_1, \dots, \delta_n)$ be
  the number of individuals sampled from species $x_i$, $i = 1, \dots,
  n$. With $x_{ij}$, $1\le j\le\delta_i$, denoting the individuals in
  taxon $x_i$, $\mathcal X^*=\{x_{ij}\}$ is the set of all sampled
  individuals, so $|\mathcal X^*|=\sum_{i=1}^n \delta_i$.

  An \emph{extended species tree} $\sigma^* = (
  \psi^*,\lambda,\delta)$ on $\mathcal X$ is a species tree
  $(\psi^*,\lambda)$ on $\mathcal X^*$ such that for each $1\le i\le
  n$ all the leaves $x_{ij}$, $1\le j\le \delta_i$ have a common
  parent in $\psi^*$. 

The \emph{pruned species tree topology} $\psi$ on
  $\mathcal X$ is obtained from $ \psi^*$ by 
  labeling the parent of the $x_{ij}$ by $x_i$
  for each $i$ with $\delta_i > 1$, 
  and then excising
  the leaves $x_{ij}$ and the pendant edges on which they lie.
\end{defn}

Note that while an extended species tree gives rise to a species tree
by the pruning process,
in an extended species tree a branch length is assigned to those edges
which become pendant in the species tree whenever there are two or more
sampled individuals in the taxon. Since our notion of a species tree in this 
paper does not have pendant edge lengths, an extended species tree thus 
carries more edge length information than the associated species tree.

Gene trees arising from the coalescent model on an extended species tree
have leaves labelled $\mathcal X_g = \{X_{11}, \dots, X_{1\delta_1},
\dots, X_{n1}, \dots, X_{n\delta_n}\}$ and are, with probability 1,
binary. One readily checks that the
probability of such a gene tree under the multispecies coalescent on
the extended species tree is exactly the same as the
probability of the gene tree under a multiple individual sampling
scheme on the species tree (with some pendant edge lengths) obtained
by pruning. Indeed, this is why we have introduced such trees. We will use 
them to easily extend results where one individual is sampled per species 
to the multiple sampling situation, in Proposition 13 and Corollary 14.

\begin{figure}[!t]
\begin{center}
\setlength{\unitlength}{0.05 mm}%
   \begin{picture}(100,1204)(0,0)

\put(-1200,0){
\includegraphics[width=.48\textwidth]{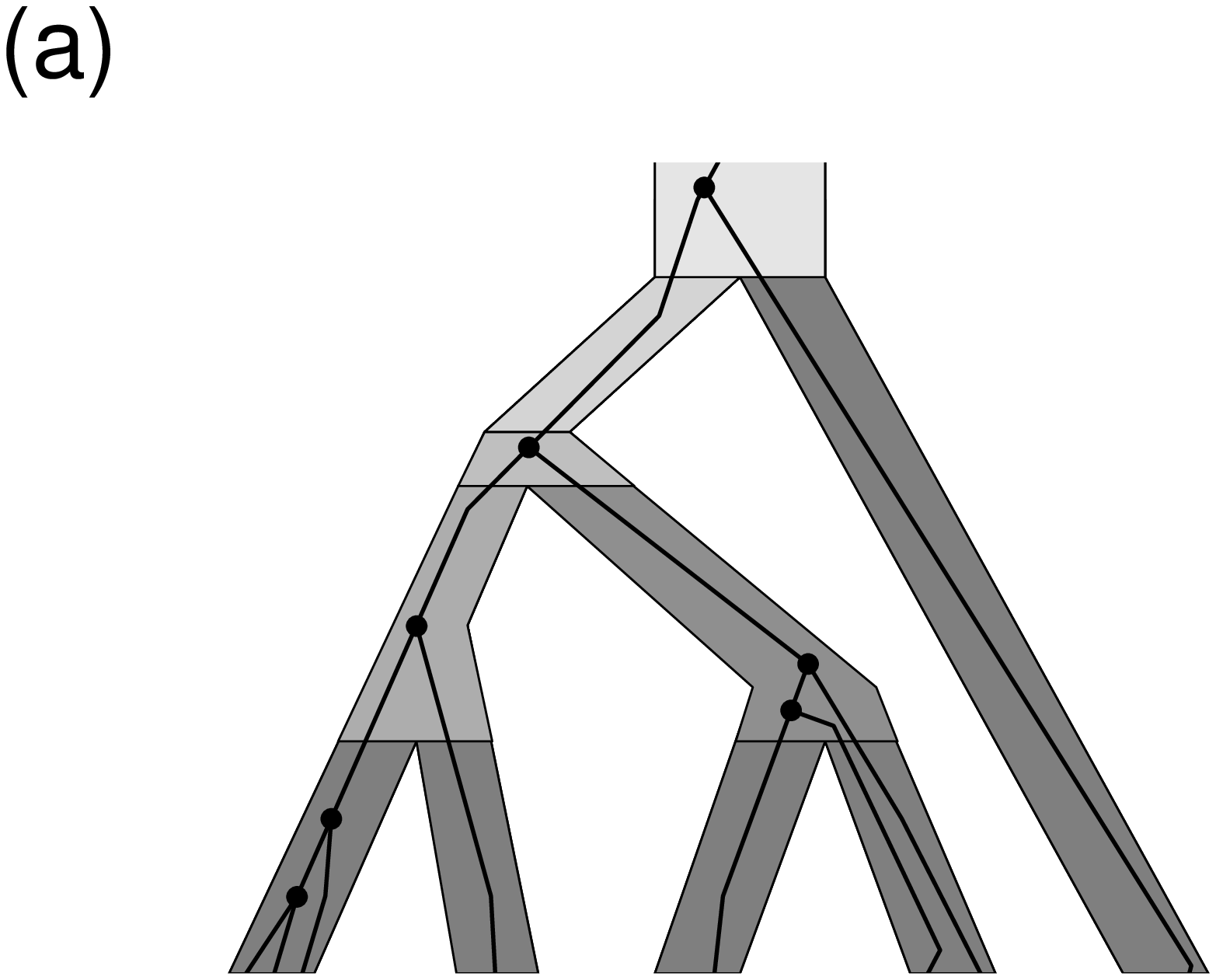}}
\put(-950,290){\fontsize{11.23}{14.07}\selectfont   \makebox(2436.0, 0.0)[l]{$x_1$ \;\;\;\;\;$x_2$ \;\;$x_3$\; \;\;\;\;$x_4$\; \;\;\;$x_5$\strut}}
\end{picture}
   \begin{picture}(100,1204)(0,0)
\put(0,0){
\includegraphics[width=.48\textwidth]{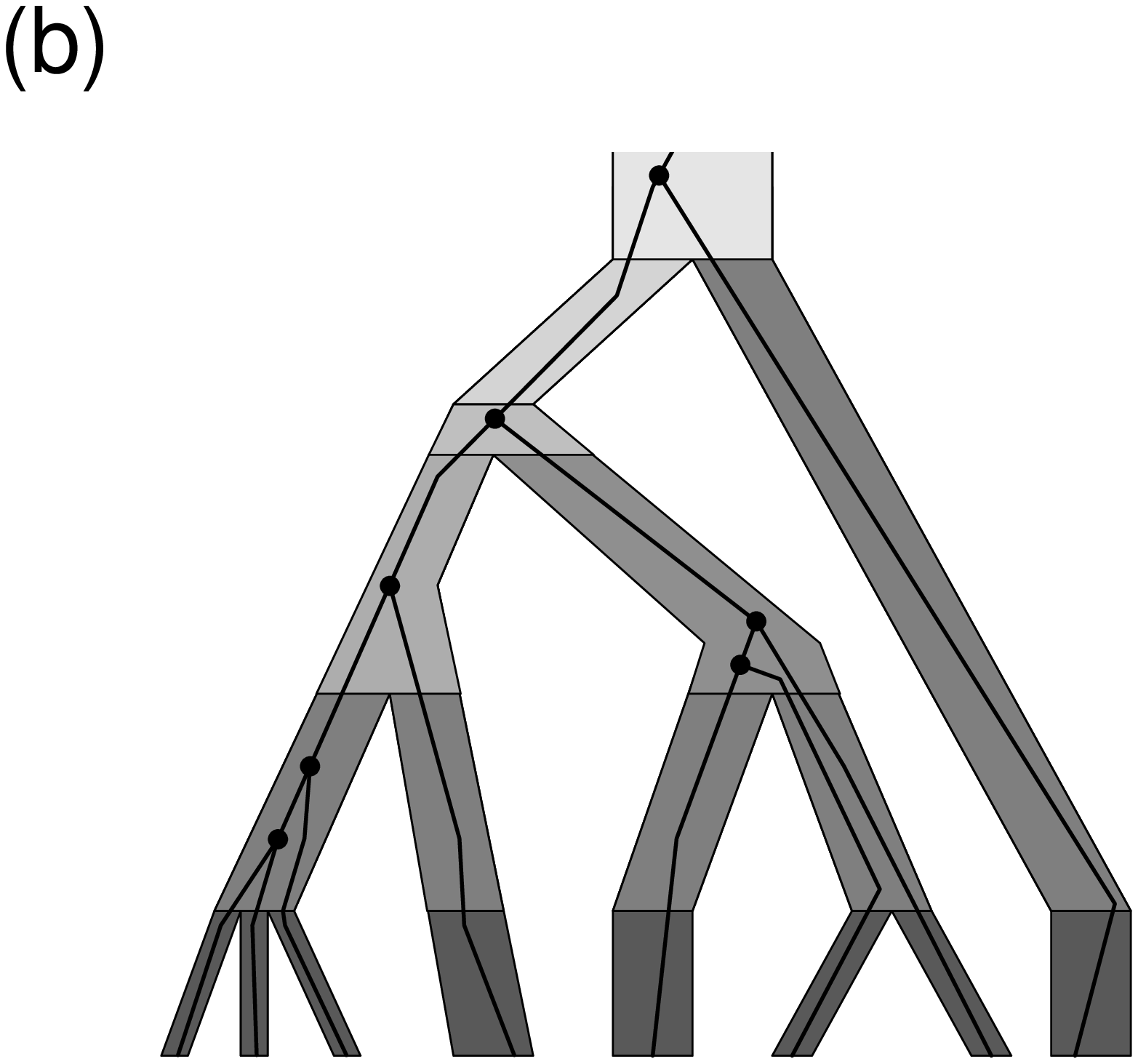}}
\put(100,150){\fontsize{11.23}{14.07}\selectfont   \makebox(2436.0, 0.0)[l]{$x_{11}\, x_{12}\, x_{13}\,x_{21}$ \;$x_{31}\, x_{41}\;\; x_{42}$ $x_{51}$\strut}}
\end{picture}
\caption{(a) A gene tree with multiple lineages sampled from several
  species within a species tree.  The taxa are $\mathcal X =
  \{x_1, x_2, x_3, x_4, x_5\}$ with $\delta = (3,1,1,2,1)$ lineages
  sampled from them. (b) The extended version of the
  species tree with taxa $\mathcal X^* = \{x_{11}, x_{12}, \dots,
  x_{51}\}$ and one lineage sampled per taxon.  Under the multispecies coalescent, the probability of
  any clade $\mathcal A_g \subset \mathcal X^*_g$ is the same 
  for both the species trees in (a) and in (b).}
\label{F:intra}
\end{center}
\end{figure}

Finally, note that by construction, for each $i=1, \dots, n$, the set $\mathcal A_i = \{x_{i1},
\dots, x_{i\delta_i}\}$ is a clade on the extended species tree.  But
of course a set $(\mathcal A_i)_g=\{X_{i1}, \dots,
X_{i\delta_i}\}$ need not be a clade on any given gene tree.

\section{Arbitrary gene tree distributions}
Though the remainder of this paper is concerned only with the gene
tree distribution arising from the multispecies coalescent model, in
this section we investigate clade probabilities for arbitrary binary gene
tree distributions. The main observation is that without special
assumptions on the gene tree distribution, the clade probabilities do
not contain enough information to recover the gene tree distribution.

\medskip

Note that every gene tree must have as clades all singleton sets of gene labels,
as well as the full set $\mathcal X_g$. We refer to these
as \emph{trivial clades}. Any other subset $\mathcal C_g\subset
\mathcal X_g$ is a clade on some gene trees, but not others.

For an arbitrary distribution of gene trees on a taxon set $\mathcal
X$, let $\PP(T)$ denote the probability of gene tree $T$. Then for
each subset $\mathcal C_g\subseteq \mathcal X_g$, the probability that
$\mathcal C_g$ is a clade on a gene tree is
 $$\PP(\mathcal C_g)= \sum_T \PP(\mathcal C_g | T) \PP(T) = \sum_T I(\mathcal C_g\in \mathcal H(T))\, \PP(T),$$
 where $I$ is the indicator function with values of 1 or 0.  Note that
 the probability of any trivial clade is therefore 1.

 We emphasize that the clade probabilities for an $n$-taxon species
 tree $\sigma$ do \emph{not} form a probability distribution. The
 presence of different clades may not be mutually exclusive events
 (for instance, if $\mathcal C_g \subset \mathcal C'_g$), and their
 probabilities do not sum to 1. 
 
\begin{prop} \label{prop:trivinv}
If $|\mathcal X|=n$, then for any distribution of binary gene trees the sum of the probabilities of all non-trivial clades is $n-2$.
\end{prop}
\begin{proof} Denoting $n$-taxon gene trees by $T$,
$$\sum_{\mathcal C_g\subset \mathcal X_g \atop\text{non-trivial}} \mathbb P(\mathcal C_g)= \sum_{\mathcal C_g\subset \mathcal X_g \atop\text{non-trivial}}  \sum_{T} I(\mathcal C_g \in \mathcal H(T))\, \mathbb P(T)=\sum _{T} \sum_{\mathcal C_g \subset \mathcal X_g\atop\text{non-trivial}} I(\mathcal C_g\in \mathcal H(T))\,  \mathbb P(T).$$
But since each binary gene tree has $n-2$ non-trivial clades, this shows that
\begin{equation}\label{E:trivinv}
\sum_{\mathcal C_g \subset \mathcal X_g\atop\text{non-trivial}} \mathbb P(\mathcal C_g)=\sum _{T} (n-2) \mathbb P(T)
=n-2.
\end{equation}
\end{proof}

\begin{thm}\label{thm:toomany}
  For an arbitrary distribution of binary gene trees on a taxon set $\mathcal
  X$ with $|\mathcal X|\ge 4$, the gene tree probabilities $\PP(T)$
  cannot be identified from the clade probabilities $\PP(\mathcal
  C_g)$.
\end{thm}
\begin{proof} The set
  $\mathcal X_g$ has $2^n-n-2$ subsets $\mathcal C_g$ with $2\le
  |\mathcal C_g|\le n-1$. Using Proposition \ref{prop:trivinv}, the
  clade probabilities can thus be specified by point in a
  $(2^n-n-3)$-dimensional vector space. However, there are
  $(2n-3)!!=1\cdot 3\cdots(2n-3)$ binary gene trees on $\mathcal X_g$, so a
  gene tree distribution is specified by a point in a
  $((2n-3)!!-1)$-dimensional vector space. But
  since $$(2n-3)!!-1>2^n-n-3$$ when $n\ge 4$, and the map from gene
  tree probabilities to clade probabilities is linear, the map is not
  invertible at any point.
\end{proof}

We note that for an arbitrary distribution on multifurcating gene tree topologies, 
the trivial invariant in Eq.~\eqref{E:trivinv} need not hold. However, the argument establishing Theorem \ref{thm:toomany}
can be modified to apply to such distributions, since the number of multifurcating trees is greater than the number of binary ones.

\section{Highly probable gene tree clades are species tree clades}\label{sec:high}

For the remainder of the paper, we assume that both gene tree probabilities
$\PP_\sigma( T)$ and clade probabilities $\PP_\sigma(\mathcal C_g)$
arise from the multispecies coalescent on a species tree
$\sigma=(\psi,\lambda)$.

\begin{thm} Let $\sigma=(\psi,\lambda)$ be a binary species tree on
  $\mathcal X$, with edge lengths $\lambda_i>\epsilon \ge 0$.  Under the
  multispecies coalescent model, suppose $\mathcal C_g\subset \mathcal
  X_g$ has clade probability $\mathbb P_{\sigma}(\mathcal
  C_g)\ge(1/3)\exp(-\epsilon)$. Then $\mathcal C$ is a clade on
  $\sigma$; that is, $\mathcal C\in \mathcal H(\sigma)$.
  
Furthermore, if $(1/3) \exp(-\epsilon)$ is replaced with any smaller
  number, this statement is no longer true for all such choices of
species trees and non-trivial clades: For any $k<(1/3) \exp(-\epsilon)$, there exists a 
species tree $\sigma$on $\mathcal X$ and a taxon set $\mathcal C\subset \mathcal X$ with 
$1<|\mathcal C|<|\mathcal X|$ such that
$\mathcal C$ is not a clade on $\sigma$, yet $\mathbb P_{\sigma}(\mathcal
  C_g)\ge k$.
 \end{thm}

\begin{proof} If $\mathcal C$ is a trivial clade, there is nothing to show, 
   so we may assume $1< |\mathcal C_g|<|\mathcal X|$.  We
  prove the contrapositive: if $\mathcal C$ is not a clade on $\psi$,
  then $\mathbb P_{\sigma}(\mathcal C_g)<(1/3)\exp(-\epsilon)$.

  Suppose $\mathcal C$ is not a clade on the species tree,
so there exist $a,b \in \mathcal C$ and $c \in \mathcal
  X\smallsetminus\mathcal C$ such that $\psi$ does not display the rooted 
  triple $((a,b),c)$. 
   Thus, the rooted triple probability satisfies
  $\mathbb P_{\sigma}( ((A,B),C ))< (1/3)
  \exp(-\epsilon)$ \cite{nei1987}.  But then $$(1/3)
  \exp(-\epsilon)>\mathbb P_{\sigma} (((A,B),C)) \ge \mathbb
  P_{\sigma}(\mathcal C_g),$$ since $((A,B),C)$ is
  displayed on every gene tree on which $\mathcal C_g$ is a clade.
  \smallskip

To establish the last claim of the theorem, we construct an example.
For any set $\mathcal C$ with $1<|\mathcal C|<|\mathcal X|$, pick some
$a\in \mathcal C$, and some $c\in \mathcal X\smallsetminus \mathcal
C$. Let $\mathcal C'=\mathcal C\smallsetminus\{a\}$. Consider a binary
species tree $\sigma$ which has a 
subtree of the form
$((a,c)\tc\delta,T_{\mathcal C'}\tc \gamma)$, where $T_{\mathcal C'}$ is
any rooted tree on $\mathcal C'$. Note then that $\mathcal C$ is not a
clade on $\sigma$.

By taking $\gamma$ to be large, the probability that the lineages from $\mathcal
C'$ coalesce below $\MRCA(\{a,c\}\cup\mathcal C')$ can be made as
close to 1 as desired. Because the probability that lineages $A$ and $C$ fail to coalesce within time $\delta$ is $\exp(-\delta)$, by also choosing $\delta\approx\epsilon$ 
the probability that three lineages (one
for $A$, one for $C$, and one for $\mathcal C_g'$) enter the ancestral
population above this MRCA can be made as close to $\exp(-\epsilon)$
as we wish.  Thus the probability that $\mathcal C_g$ will be a clade
on a gene tree can be made as close to $(1/3)\exp(-\epsilon)$ as we
wish by taking the branch above this subtree to be long.
\end{proof}

Setting $\epsilon=0$ yields
Corollary \ref{cor:third}.

\begin{cor} \label{cor:third} Let $\sigma$ be a binary species tree on
  taxa $\mathcal X$, with positive edge lengths.  Under the
  multispecies coalescent model, suppose $\mathcal C\subset \mathcal
  X$ is such that $\mathbb P_{\sigma}(\mathcal C_g)\ge1/3.$ Then
  $\mathcal C$ is a clade on $\sigma$.

  Furthermore, this statement is no longer true for $1< |\mathcal C|<
  |\mathcal X|$ if $1/3$ is replaced with any smaller number.
\end{cor}

If the species tree is not binary, a slightly weaker result holds.

\begin{thm}
  Suppose the species tree $\sigma$ is not necessarily binary, and
  $\mathcal C\subset \mathcal X$ is such that $\mathbb
  P_{\sigma}(\mathcal C_g)>1/3$. Then $\mathcal C$ is a clade on
  $\sigma$.

 Furthermore, this statement is no longer true for $1<|\mathcal
 C|<|\mathcal X|$ if $1/3$ is replaced with any smaller number.
\end{thm}

\begin{proof}
To show $\mathcal C$ is a clade,  we suppose $c\in \mathcal X \smallsetminus \mathcal C$ and demonstrate that $c\notin \desc(\MRCA(\mathcal C))$.
Choose $a,b\in \mathcal C$ such that $\MRCA(\mathcal C)=\MRCA(\{a,b\})$.
Note that $\mathbb P_{\sigma}( ((A,B),C) )\ge \mathbb
P_{\sigma}(\mathcal C_g)>1/3$ implies
that the rooted triple $((a,b),c)$ is displayed on
$\sigma$. Thus $c \notin \desc(\MRCA(\mathcal C))$. 
  
That $1/3$ cannot be replaced with a smaller number is a consequence
of Corollary \ref{cor:third}. 
\end{proof}

\section{Clade invariants}\label{sec:inv}

A \emph{clade invariant} for a species tree topology  is a
polynomial in the probabilities of clades on gene trees that vanishes for all edge
length assignments to the species  tree.  
More completely, a clade invariant associated to an $n$-taxon species tree topology $\psi$ is a multivariate polynomial in $2^{n}-n-2$ indeterminates (one for every non-trivial clade) 
which evaluates to zero at any
vector of clade probabilities $\PP_\sigma(\mathcal C_g)$ arising from $\sigma=(\psi, \lambda)$, regardless of the values of $\lambda$. 

Proposition
\ref{prop:trivinv} gives an example of a clade invariant for binary gene trees
 that, in
addition, is independent of all features of $\psi$ except the number
of taxa:
$$\sum_{\mathcal C_g\subset \mathcal X_g\atop \text{non-trivial}}\mathbb P_\sigma(\mathcal C_g) - (n-2)=0.$$ 
We call this the \emph{trivial invariant}, and emphasize it is
satisfied by clade probabilities from any species tree on $\mathcal
X$.

Clade invariants can be computed for small trees using computational
algebra software, such as Singular \cite{GPS09}. For each edge length
$\lambda_i$, one sets $\Lambda_i=\exp(-\lambda_i)$, and then expresses
the clade probabilities as multivariate polynomials in the
$\Lambda_i$. Gr\"obner basis methods for variable elimination then
allow one to determine generators of the polynomial ideal of all clade
invariants.  Such computations were useful in formulating the
general construction of certain linear invariants given below.   The 
existence of these clade invariants forms the basis our proof
of species tree topology identifiability in Section \ref{sec:id}.

\begin{thm}\label{thm:cladeinv}
  Let $\mathcal A\subsetneq \mathcal X$ be a subset of taxa with at
  least two elements, and $\mathcal C\subseteq \mathcal
  X\smallsetminus\mathcal A$ a non-empty set of taxa not in $\mathcal
  A$. For distinct $a,b\in \mathcal A$, let $\mathcal A'=\mathcal
  A\smallsetminus\{a,b\}$.  Then if $\mathcal A$ is a clade on
  $\sigma$,
\begin{equation}
  \left( \sum_{\mathcal S\subseteq \mathcal A'} \PP_\sigma(\mathcal S_g\cup\{A\}\cup \mathcal C_g)\right )
  - \left ( \sum_{\mathcal S\subseteq \mathcal A'}\PP_\sigma(\mathcal S_g\cup\{B\}\cup \mathcal C_g)\right )=0.
\label{eq:cladeinv}
\end{equation}
\end{thm}

\smallskip

 We note that this theorem applies to any species tree, including
 non-binary ones.   Moreover, since a non-binary species tree $\sigma$ can
  be thought of as any of its binary resolutions with length 0
  assigned to any introduced edges, the clade probabilities arising
  from such a $\sigma$ will satisfy the polynomials of the theorem for every
  binary resolution. Thus in the statement of the theorem the phrase
  `if $\mathcal A$ is a clade on 
  $\sigma$' 
 can be replaced with `if
  $\mathcal A$ is a clade on a binary resolution of 
  $\sigma$.'
 
\smallskip

For the proof, it is useful to have the notion of compatible clades:
\begin{defn}
  Two clades, $\mathcal A_g$ and $\mathcal B_g$ are \emph{compatible}
  if $\mathcal A_g \cap \mathcal B_g = \emptyset$, $\mathcal
  A_g\subseteq \mathcal B_g$, or $\mathcal B_g \subseteq \mathcal
  A_g$.
\end{defn}

 If a  clade $\mathcal A_g$ is on a gene
tree $T$, then all other clades appearing on $T$ must be compatible with $\mathcal A_g$.

\smallskip

The proof of Theorem \ref{thm:cladeinv} uses partitions of subsets of
the taxon set $\mathcal X$ that occur as follows: Consider an internal
node $v$ of $\sigma$, and let $\mathcal A=\desc(v)$. Then in a
realization of the coalescent process on $\sigma$, some of the
lineages of genes in $\mathcal A_g$ may coalesce below $v$, so that
there are $|\mathcal A|$ or fewer lineages at $v$. Each such lineage
determines a subset of $\mathcal A_g$, namely its descendants, and
hence the set of lineages determines a partition of $\mathcal A$. 

As an example, consider the species
tree in Fig.~\ref{F:partition}.  For the set $\mathcal A =
\{a,b,c\}$, the partition at $\MRCA(\mathcal A)$ in both subfigures is
$\{\{a\}, \{b\}, \{c\} \}$.  Note that the partition of such a set
$\mathcal A$ is not affected by any coalescent events occurring in
the MRCA population, but only by those below.  
The only
other partition of $\mathcal A$ possible for this species tree is $\{
\{a,b\}, \{c\}\}$.  For the set $\mathcal A = \{a,b,c,d\}$, the
partition at $\MRCA(\mathcal A)$ in Fig.~\ref{F:partition}a is $\{ \{a\}, \{b,c\},
\{d\}\}$, and 
in Fig.~\ref{F:partition}b is $\{\{a,b,c\}, \{d\}\}$.

\smallskip

\begin{proof}[Proof of Theorem \ref{thm:cladeinv}] 

  Suppose $\mathcal A$ is a clade on
  $\psi$, with $v= \MRCA(\mathcal A)$.  Letting
  $\pi(\mathcal{A}) = \mathcal A_1 | \cdots | \mathcal A_k$ denote a
  partition of $\mathcal A$, we also use $\pi(\mathcal A)$ to denote
  the event that the coalescent process on $\sigma$ produces lineages
  at $v$ defining this partition. 
  
   We will condition on this event: Specifically, recalling the notion
   of a coalescent history from \cite{degnan2005},    \begin{equation}\label{eq:joint}
  \PP_{\sigma} (\pi (\mathcal A)) = \sum_T \sum_{\text{history } h_T,
    \atop{h_T \text{ consistent} \atop{\text{with }
        \pi(\mathcal A)}}} \PP_{\sigma} (T, h_T).
  \end{equation}
  For
$\mathcal B \subset \mathcal X$ the joint probability $\PP_{\sigma} (\mathcal B_g, \pi (\mathcal A) )$  
is computed similarly, by restricting the outer sum on the right side of Eq.~\eqref{eq:joint} to those gene trees that have clade $\mathcal B_g$.
 Then  
     $$\PP_{\sigma} (\mathcal B_g \, | \, \pi (\mathcal A) )= \frac{\PP_{\sigma} (\mathcal B_g, \pi (\mathcal A) )}{\PP_{\sigma} (\pi (\mathcal A)) },$$ 
       and by the law of total probability, we have the clade probability
 \begin{align*}
  \PP_{\sigma}(\mathcal B_g) &= \sum_{\pi(\mathcal{A})} \PP_{\sigma}(\mathcal B_g \, | \, \pi(\mathcal A)) \, \PP_{\sigma}(\pi(\mathcal A)).
\end{align*}
Thus, to establish Eq.~\eqref{eq:cladeinv}, it is enough to show
that 
\begin{equation}
  \left( \sum_{\mathcal S\subseteq \mathcal A'} 
  \PP_\sigma\left(\mathcal S_g\cup\{A\}\cup \mathcal C_g \, | \, \pi(\mathcal A)\right)\right )
  - \left ( \sum_{\mathcal S\subseteq \mathcal A'}
  \PP_\sigma\left(\mathcal S_g\cup\{B\}\cup \mathcal C_g \, | \, \pi(\mathcal A)\right)\right )=0
\label{eq:cladeinv22}
\end{equation}
holds for all choices of partition $\pi(\mathcal A)$.

To establish Eq.~\eqref{eq:cladeinv22}, we show that
non-zero terms cancel pairwise. However, which terms cancel 
depends on the partition, so for the remainder of the
argument we fix $\pi(\mathcal A)$,
and assume the partition sets are indexed so that $a\in \mathcal A_1$. 

Note first that if $b\in\mathcal A_1$ as well, then we are
conditioning on an event that requires that the $A$ and $B$ lineages
have coalesced into one below $v$.  Thus, any clade on
a gene tree that includes $A$ and $\mathcal C_g$ must include $B$,
because we have assumed that $\mathcal C$ is non-empty.  Similarly, any
clade that includes $B$ and $\mathcal C_g$ must include $A$.
Therefore, all probabilities in Eq.~\eqref{eq:cladeinv22} are
zero, so the equation holds.

Otherwise, assume $b\in \mathcal A_2$.  We wish to give a bijective
correspondence between non-zero clade probabilities in the first sum
in Eq.~\eqref{eq:cladeinv22} and equal clade probabilities in the
second sum, with the correspondence dependent on the partition 
$\pi (\mathcal A)$.
That is, we wish to show that for each $\mathcal S_1 \subset \mathcal
A^{\prime}$, there is a corresponding $\mathcal S_2 \subset \mathcal
A^{\prime}$ such that
\begin{equation}\label{eq:equality}
  \mathbb P_{\sigma}[(\mathcal S_1)_g \cup \{A\} \cup \mathcal C_g \, | \, 
  \pi(\mathcal A)] = \mathbb P_{\sigma}[(\mathcal S_2)_g \cup \{B\} \cup \mathcal C_g 
  \, | \, \pi(\mathcal A)].
\end{equation}
Consider first the case when 
$(\mathcal S_1)_g\cup \{A\}\cup \mathcal C_g$ is compatible
with the clades $(\mathcal A_1)_g, \dots, (\mathcal A_k)_g$.
Because $\mathcal C$ is non-empty,
this occurs exactly when
$\mathcal S_1\cup\{a\}$ is the union of some of the $\mathcal A_i$.  Thus we have
$$
\mathcal S_1\cup\{a\}=\mathcal A_1\sqcup \bigsqcup_j \mathcal A_{i_j},
$$
for some $i_j$, with all unions here disjoint. Moreover, since
$b\notin \mathcal S_1\cup\{a\}$, $\mathcal A_2$ does not appear in
this expression.  We therefore define $\mathcal S_2$ by the
expresssion of disjoint unions
$$
\mathcal S_2\cup\{b\}=\mathcal A_2\sqcup \bigsqcup_j \mathcal A_{i_j}.
$$
Eq.~\eqref{eq:equality} then holds, since for the coalescent
process on $\sigma$ above $v$ the lineages corresponding to $\mathcal
A_1$ and $\mathcal A_2$ are exchangeable.
This gives us a bijection between $\mathcal S_1 \subset \mathcal A'$ and 
$\mathcal S_2 \subset \mathcal A'$ for which either (and hence both) of
the probabilities in Eq.~\eqref{eq:equality} are non-zero.

For all other $\mathcal S_1$, $\mathcal S_2$, the sets $\mathcal S_1 \cup \{A\} \cup \mathcal C_g$
and $\mathcal S_2 \cup \{B\} \cup \mathcal C_g$ are not compatible with
$\pi (\mathcal A)$, and hence these probabilities are zero.
\end{proof}

As a simple corollary, we immediately obtain what we call `cherry-swapping'
invariants, 
which express that
the probability of any clade containing exactly one
taxon of a 2-clade on the species tree is unchanged when that taxon is
swapped out for the other taxon in the 2-clade.  

\begin{cor} (Cherry-swapping invariants)\label{cor:cherryswap} Suppose
  $\{a,b\}$ is a 2-clade on a species tree with taxa $\mathcal
  X$. Then for any $\mathcal C\subseteq \mathcal
  X\smallsetminus\{a,b\}$,
$$\PP_\sigma(\{A\}\cup \mathcal C_g)-\PP_\sigma(\{B\}\cup \mathcal C_g)=0.$$
\end{cor}

\medskip

To illustrate Theorem \ref{thm:cladeinv},
we consider next all species tree topologies on $5$ or fewer taxa, and discuss
invariants produced by this construction. 
For notational ease, we denote gene tree clades
by juxtaposition of labels, rather than by sets, so,  for instance
$\{A,B,D,E\}$ will be denoted $ABDE$.   Our focus is
on those invariants associated to $3$- and $4$-clades, and
we do not explicitly list cherry-swapping invariants except for the 3-taxon tree.

\begin{example}
For the species  tree topology $\psi=((a,b),c)$, the cherry-swapping invariant,
$$\PP_\sigma(AC)-\PP_\sigma(BC)=0,$$
is the only one produced by Theorem \ref{thm:cladeinv}.
\end{example}

\begin{example} For the $4$-taxon caterpillar tree topology $\psi = (((a,b),c),d)$, in
  addition to the three cherry-swapping invariants, we find for
  $\mathcal A=\{a,b,c\}$ the invariants
\begin{align*}
&\left( \PP_\sigma(AD)+\PP_\sigma(ABD)\right )-\left( \PP_\sigma(CD)+\PP_\sigma(BCD)\right )=0, \ (\text{for }\mathcal A' = \{b\})\\
&\left( \PP_\sigma(BD)+\PP_\sigma(ABD)\right )-\left( \PP_\sigma(CD)+\PP_\sigma(ACD)\right )=0, \ (\text{for }\mathcal A' = \{a\})\\
&\left( \PP_\sigma(AD)+\PP_\sigma(ACD)\right )-\left( \PP_\sigma(BD)+\PP_\sigma(BCD)\right )=0, \ (\text{for }\mathcal A' = \{c\}).
\end{align*}
We note that there are relations between these: the second invariant
is obtained from the first by a cherry-swapping move, and the third 
is the sum of two cherry-swapping
invariants.

For the 4-taxon balanced tree topology, $\psi = (((a,b),(c,d))$, only the six
cherry-swapping invariants are obtained.
\end{example}

\begin{example}
  
If $\psi$ is either the $5$-taxon caterpillar tree topology $((((a,b),c),d),e)$, 
  or the balanced tree topology $(((a,b),c),(d,e))$, consider $\mathcal A=\{a,b,c\}$.  
  
  Then for $\mathcal A' = \{b\}$, we obtain
  for various choices of $\mathcal C$,
\begin{equation}\label{E:5cat1}
\left( \PP_\sigma(AD) +\PP_\sigma(ABD)\right )-\left ( \PP_\sigma(CD) + \PP_\sigma(BCD) \right )=0,
\end{equation}
\begin{equation}\label{E:5cat2}
\left( \PP_\sigma(AE) +\PP_\sigma(ABE)\right )-\left ( \PP_\sigma(CE) + \PP_\sigma(BCE) \right )=0,
\end{equation}
\begin{equation}\label{E:5cat3}
\left( \PP_\sigma(ADE) +\PP_\sigma(ABDE)\right )-\left ( \PP_\sigma(CDE) + \PP_\sigma(BCDE) \right )=0.
\end{equation}

Note that for the balanced species tree, Eq.~\eqref{E:5cat2} follows from Eq.~\eqref{E:5cat1} by cherry swapping $D$ and $E$.
However, for the caterpillar species tree, Eqs.~\eqref{E:5cat1} and \eqref{E:5cat2} are not related by a cherry swap.

For $\mathcal A' = \{a\}$ with singleton $\mathcal C$ we obtain
Eqs.~\eqref{E:5cat1}--\eqref{E:5cat3}
again, up to cherry-swapping lineages $A$ and $B$.

For $\mathcal A' = \{c\}$ with singleton $\mathcal C$, we obtain
invariants such as
\begin{equation}\label{E:5cat4}
\left( \PP_\sigma(AD) +\PP_\sigma(ACD)\right )-\left ( \PP_\sigma(BD) + \PP_\sigma(BCD) \right )=0,
\end{equation}
but Eq.~\eqref{E:5cat4} is simply the sum of two cherry-swapping invariants for the
cherry $\{a,b\}$, with $\mathcal C=\{d\}$ and $\{c,d\}$. 
In general,
if the taxa in $\mathcal A\smallsetminus A'$ span a smaller clade
than $\mathcal A$, the invariant produced will be a sum of invariants
for the smaller clade. Indeed, this phenomenon occurred above, for the 4-taxon caterpillar.

\end{example}

\begin{example}
  For the $5$-taxon caterpillar topology $\psi = ((((a,b),c),d),e)$, taking $\mathcal A=\{a,b,c,d\}$ 
  and using $\mathcal A' = \{b,c\}$ and 
  $\mathcal A' = \{a, b\}$,  we obtain two
  invariants:
\begin{multline*}\left( \PP_\sigma(AE) +\PP_\sigma(ABE)+\PP_\sigma(ACE)+\PP_\sigma(ABCE)\right )\\
-\left( \PP_\sigma(DE) +\PP_\sigma(BDE)+\PP_\sigma(CDE)+\PP_\sigma(BCDE)\right )=0, \end{multline*}
and
\begin{multline*}\left( \PP_\sigma(CE) +\PP_\sigma(ACE)+\PP_\sigma(BCE)+\PP_\sigma(ABCE)\right )\\
-\left( \PP_\sigma(DE) +\PP_\sigma(ADE)+\PP_\sigma(BDE)+\PP_\sigma(ABDE)\right )=0.\end{multline*}
Other choices of $\mathcal A'$ 
give only invariants in the space spanned by those previously discussed.
\end{example}

\begin{example}
For the $5$-taxon pseudo-caterpillar tree topology 
$\psi = (((a,b),(d,e)),c)$, taking $\mathcal A=\{a,b,d,e\}$ we obtain
\begin{multline}\label{E:ps}
\left( \PP_\sigma(AC) +\PP_\sigma(ABC)+\PP_\sigma(ACE)+\PP_\sigma(ABCE)\right )\\
-\left( \PP_\sigma(CD) +\PP_\sigma(BCD)+\PP_\sigma(CDE)+\PP_\sigma(BCDE)\right )=0,
\end{multline}
and three other invariants that can also be obtained by cherry
swapping from Eq.~\eqref{E:ps}.  Since $\PP_\sigma(ACE)=\PP_\sigma(BCD)$ by
cherry swapping, two of the eight terms can be cancelled.
\end{example}

\begin{rem}
  All the linear invariants above for 3-, 4-, and 5-taxon trees are, of course,
  among those that can be found computationally. Gr\"obner basis
  calculations do not necessarily produce exactly these, but by
  cherry-swapping and taking suitable linear combinations of computed
  linear invariants, all of these appear.  However, at least for trees
  on four and five taxa, there are additional linear invariants beyond the
  ones of Theorem \ref{thm:cladeinv}.  We give these in
  Appendix \ref{app:moreinv}, as it would be interesting to have
  non-computational means of obtaining them, as well as the higher degree invariants.

 \end{rem}

\section{Identifying  clades}\label{sec:id}

Suppose we are given the clade probabilities $\{ \PP_\sigma( \mathcal C_g)\} $
arising from the multispecies coalescent on an unknown species tree
$\sigma$, and we wish to know if $\sigma$ displays a particular
clade. By the results of Section \ref{sec:high}, high probability
may identify some clades on $\sigma$. However, it remains to be
seen how one might identify clades on $\sigma$ that have lower
probability of occurring on gene trees as a result of high levels of incomplete lineage sorting.

From Section \ref{sec:inv} we know that if $\mathcal A$ is a clade on
$\sigma$ then for every non-empty subset $\mathcal C\subseteq \mathcal
X\smallsetminus\mathcal A$, and every $a,b\in \mathcal A$, the linear
invariant associated to $\mathcal A$, $\mathcal C$, $a$, and $b$ 
vanishes. For these invariants to be useful for identifying clades,
however, we must also know that if $\sigma$ does not display the clade
$\mathcal A$, then one of these invariants 
does
not vanish.

\begin{lem}\label{prop:taxonSwapFromClade}
  Suppose $\mathcal A$ is a non-trivial clade on a species tree $\sigma$, and $a
  \in \mathcal A$ and $b \in \mathcal X\smallsetminus \mathcal A$.
  Let $\mathcal{B}$ denote the set obtained by replacing $a$ with $b$
  in $\mathcal A$, that is, $\mathcal{B} = \big( \mathcal A
  \smallsetminus \{a\} \big) \cup \{b\}$.  Then $\PP_\sigma ( \mathcal
  A_g) > \PP_\sigma( \mathcal B_g)$.
\end{lem}

\begin{proof}
  Let $v=\MRCA(\mathcal A \cup \{b\})$ on $\sigma$. Let $\mathcal
  A'=\mathcal A\smallsetminus\{a\}$, so $\mathcal A=\{a\}\cup \mathcal A'$
  and $\mathcal B=\{b\}\cup \mathcal A'$.

  Then, using phrases such as `$A$ coalesces above $v$' to mean the
  lineage of $A$ first coalesces with any other gene lineage in a
  population in the species tree above the node $v$,
  \begin{align}
 \PP_{\sigma}(\mathcal A_g) &= \PP_\sigma(\{A\}\cup \mathcal A'_g) \\
&= \PP_\sigma(\{A\}\cup  \mathcal A_g' \text{ and $A$ coalesces below $v$})\, +
\\ &\ \ \ \ \ \ \ \ \ \ \ \ \ \ \ \
 \PP_\sigma(\{A\}\cup  \mathcal A_g' \text{ and $A$ coalesces above $v$}) \notag\\
 &>\PP_\sigma(\{A\}\cup  \mathcal A_g' \text{ and $A$ coalesces above $v$})\\
 &\ge \PP_\sigma(\{A\}\cup  \mathcal A_g' \text{,  $A$ coalesces above $v$, and $B$ coalesces above $v$})\label{eq:acoal}\\
&= \PP_\sigma(\{B\}\cup  \mathcal A_g' \text{, $A$ coalesces above $v$, and $B$ coalesces above $v$})\label{eq:bcoal}\\
&= \PP_\sigma(\{B\}\cup  \mathcal A_g')\\
&= \PP_{\sigma}( \mathcal B_g).
\end{align}
The equality between lines \eqref{eq:acoal} and \eqref{eq:bcoal} is
due to exchangeability of lineages; given any sequence of coalescences in
 the event of line  \eqref{eq:acoal}, there is an equally probable sequence of coalescences in
the event of line \eqref{eq:bcoal} in which $B$ coalesces
in $A$'s place to form $\{B\}\cup\mathcal A_g'$ instead of $\{A\}\cup\mathcal A_g'$ .
Thus, $\PP_\sigma(\mathcal A_g) >\PP_\sigma( \mathcal B_g).$
\end{proof}

\begin{rem}
One might wish to extend the above result to sets obtained by 
replacing $k$ elements in a clade with $k$ elements outside it.
However, simple examples show that this is impossible. For instance,
if $\sigma=((a,b)\tc x,c)\tc y,(d,e)\tc z)$ where $z$ is large and both $x$ and $y$ are small,
then one can have $\PP_{\sigma}(ABC)<\PP_{\sigma}(CDE)$.  For 
example, if $(x,y,z) = (0.05,0.05,2.0)$, then the highest probability 
clades are $DE$, $AB$, $AC$, $BC$, $CDE$, and $ABC$ with probabilities 
$0.889$, $0.269$, $0.220$, $0.220$, $0.194$, and $0.188$, respectively 
(computed by {\tt COAL} \cite{degnan2005}).  Thus for these branch lengths, 
we have $\PP_{\sigma}(ABC) < \PP_{\sigma}(CDE)$, and the greedy 
strategy of accepting the most probable clades one-at-a-time returns the 
non-matching tree $((a,b),(c,(d,e)))$.

The same example shows that for a set $\mathcal C \subset \mathcal X$ there can exist $y \in \mathcal C$ such that for all $x\in \mathcal X\smallsetminus \mathcal C$, $\PP_{\sigma}(\mathcal C_g) > \PP_\sigma( (\mathcal C\smallsetminus\{y\}) \cup \{x\})$ and yet $\mathcal C$ is not a clade.  In this example, $\{c,d,e\}$ is not a clade on the species tree, yet $CDE$ is more probable than $ADE$ or $BDE$ on gene trees.
\end{rem}

Lemma \ref{prop:taxonSwapFromClade} allows us to show that if all
cherry swapping invariants are satisfied for a particular candidate
2-clade, it is in fact a 2-clade on the species tree.

\begin{prop}\label{prop:2cladeSwapping}
  (Clade probabilities determine species 2-clades.) For
  $\sigma=(\psi,\lambda)$ an $n$-taxon binary species tree on a set of
  taxa $\mathcal X$, the $2$-clades of $\psi$ are identifiable from
  clade probabilities.  In particular, for any $a, b \in \mathcal X$,
  $\{a,b\}$ is a clade on $\psi$ if, and only if, for every $\mathcal
  D \subseteq \mathcal X \smallsetminus \{a,b\}$, $\PP_\sigma (\{A\}\cup
  \mathcal D_g) = \PP_\sigma(\{B\}\cup \mathcal D_g)$.
\end{prop}

\begin{proof}
  If $\{a,b\}$ is a clade on $\sigma$, then by Corollary
  \ref{cor:cherryswap}, $\PP(\{A\}\cup \mathcal D_g)= \PP(\{B\}\cup
  \mathcal D_g)$ for any taxon set $\mathcal D$ not containing $a$ or
  $b$.

  Suppose now that $\{a,b\}$ is not a clade on $\psi$.  Then, because
  $\sigma$ is binary, at least one of $a$ or $b$ (let us say $a$) is in a
  non-trivial clade $\mathcal C$ on $\psi$ that excludes the other.
  Let $\mathcal D=\mathcal C\smallsetminus\{a\}$.

  By Lemma \ref{prop:taxonSwapFromClade}, $\PP_\sigma(\{A\}\cup
  \mathcal D_g)\ne \PP_\sigma(\{B\}\cup \mathcal D_g)$.
\end{proof}

For clades of more than two taxa on a species trees, we obtain a
slightly weaker result: As long as the edge length vector $\lambda$ does not lie in a set of measure zero, then the clades on
the species tree can be identified. The first step toward this result
is the following.

\begin{lem}\label{lem:nonzero}
  Let $\psi$ be a species tree topology on $\mathcal X$, and $\mathcal
  X=\mathcal A\sqcup \mathcal D$ a disjoint union of non-empty
  subsets with $|\mathcal A| \ge 2$. Then if $\mathcal A$ is not a clade on $\psi$ and
  $\MRCA(\mathcal A)$ is a binary node, then there exists some
  $\mathcal C\subseteq\mathcal D$, $a,b\in\mathcal A$, and some choice
  of edge lengths $\lambda$ such that the corresponding clade
  invariant of Theorem \ref{thm:cladeinv} does not vanish on the clade
  probabilities arising under the multispecies coalescent on
  $\sigma=(\psi,\lambda)$.
\end{lem}

\begin{proof} Suppose $\mathcal A$ is not a clade on $\psi$.  Let
  $v=\MRCA(\mathcal A)$ on $\psi$, so $\mathcal
  E=\desc(v)\smallsetminus \mathcal A$ is non-empty. One or both
  children nodes $w_1,w_2$ of $v$ have an element of $\mathcal E$ as a
  descendant, so we may assume $\mathcal C=\desc(w_1)\cap \mathcal E$
  is non-empty. Let $a\in \mathcal A\cap\desc(w_1)$, $b\in \mathcal
  A\cap\desc(w_2)$.  Consider the clade invariant of Theorem
  \ref{thm:cladeinv} associated to $\mathcal A$, $\mathcal C$, $a$,
  $b$.

  We next give edge lengths $\lambda$ for which this invariant will
  not vanish at the clade probabilities arising from the multispecies
  coalescent on $\sigma=(\psi,\lambda)$.  Let all internal edges of
  $\psi$ below $v$ have length (near) 0 except the edge $(v,w_1)$
  which is assigned length (near) $\infty$. Lengths of edges above $v$
  can be fixed at any finite non-zero values.

  With these assignments, the only partition of $\desc(v)$ according
  to lineages at $v$ that appears with non-negligable probability is
  that with $\desc(w_1)$ forming one partition set, while all elements
  of $\desc(w_2)$ are in singleton sets. But since $\mathcal
  C\subset\desc(w_1)$, $a\in\desc(w_1)$ and $b$ is in a singleton set,
  the only clades that can result with non-negligible probability
  that contain both $B$ and elements of $\mathcal C_g$ must also
  contain $A$. Thus all the clades appearing in the second term of
  Eq.~\eqref{eq:cladeinv} have probability arbitrarily close to 0. However, the clade
  $(\desc(w_1))_g$ appears in the first term and has non-negligible
  probability. Thus Eq.~\eqref{eq:cladeinv} is violated.
\end{proof}

\begin{thm}\label{thm:meas0}
  Let $\psi$ be a rooted binary species tree topology on $\mathcal X$,
  where $\mathcal X=\mathcal A\sqcup \mathcal D$ is a disjoint union
  of non-empty subsets. If $\mathcal A$ is not a clade on $\psi$ then
  for all choices of edge lengths $\lambda$ except those in some set
  of measure zero there exists some $\mathcal C\subseteq\mathcal D$,
  $a,b\in\mathcal A$, such that the corresponding clade invariant does
  not vanish on the clade probabilities arising under the multispecies
  coalescent on $\sigma=(\psi,\lambda)$.
\end{thm}

\begin{proof} The clade probabilities arising from $\sigma$ can be
  expressed as polynomials in the exponentials of the negatives of the
  interior edge lengths. By Lemma \ref{lem:nonzero}, there is an
  invariant which, when composed with this polynomial map, does not
  vanish at some point in the space 
  $(0,1]^{n-2}$
  of these
  exponentials. But since this composition is a polynomial, its
  non-vanishing at some point implies the set where it vanishes has
  measure zero in 
  $(0,1]^{n-2}$.
  Mapping this set to interior edge
  lengths by $-\log(x)$ shows the set of edge lengths for which the
  invariant vanishes has measure zero.
\end{proof}

Since, except for a negligible set of edge length parameters, whether
a species tree has a particular clade can be tested by examining
clade probabilities, one can similarly determine the full species
tree topology.

\begin{cor}\label{cor:genid}
  Let $\psi$ be a rooted binary species tree topology on $\mathcal
  X$. For generic choices of edge lengths $\lambda$, $\psi$ can be
  identified from the probabilities of clades under the multispecies
  coalescent on $\sigma=(\psi,\lambda)$.
\end{cor}

\begin{proof}
  For any subset of taxa $\mathcal A\subset \mathcal X$, if we find
  any invariant given by Theorem \ref{thm:cladeinv} that fails to vanish on
  the clade probabilities for $\sigma=(\psi,\lambda)$, then $\mathcal
  A$ is not a clade on $\psi$. If all such invariants vanish, then by
  Theorem \ref{thm:meas0}, either $\mathcal A$ is a clade on $\psi$,
  or $\lambda$ lies in a set of measure zero (which is dependent on
  $\mathcal A, \mathcal C,a,$ and $b$ used in defining the invariant).

  Thus, considering all proper subsets $\mathcal A$ of $\mathcal X$, we
  can determine all clades, unless the edge lengths $\lambda$ lie in a
  set of measure zero (the finite union of sets of measure zero for
  each invariant.)

  Finally, the clades of $\psi$ determine $\psi$.
\end{proof}

\begin{rem}
  If one considers a non-binary species tree to be specified by the
  binary tree topologies of some resolution along with the assignment
  of edge length 0 to any introduced edges, then both Theorem
  \ref{thm:meas0} and Corollary \ref{cor:genid} still apply. Indeed,
  the special choices of some 0 edge lengths form a set of Lebesgue measure
  zero in the full set of possible edge lengths, so regardless of
  whether such trees can be identified, the statements remain valid.
\end{rem}

\smallskip

A particular feature of non-binary species  trees that is identifiable is a
\emph{$k$-cherry},  a set of $k\ge 2$ leaves $\{x_1 \dots, x_k\} \in
\mathcal X$ that both share a common parent node and form a clade.
This will prove useful for identifying the extended species trees defined in
Section \ref{sec:defs}, which describes the sampling of multiple individuals per taxon.

\begin{prop}\label{prop:kcherry}
  (Clade probabilities determine extended species tree $k$-cherries.)
  Let $\sigma^*=(\psi^*,\lambda,\delta)$ be an extended species tree
  on $\mathcal X$ for which the pruned species tree $\psi$ is
  binary. Then the $k$-cherries of $\psi^*$ are identifiable from gene
  clade probabilities from the multispecies coalescent on $\sigma^*$
  for all choices of edge lengths $\lambda$ outside a set of
  measure zero.  
  
  In particular, $\{x_{i_1j_1}, \dots, x_{i_kj_k}\}\subseteq \mathcal
  X^*$ is a $k$-cherry on $\psi^*$ if, and only if, it is a maximal
  subset of $\mathcal X^*$ such that for every $1\le l<m\le k$ and every $y
  \in \mathcal X^*\smallsetminus\{x_{i_lj_l},x_{i_m j_m}\}$, $$\PP_\sigma(\{ X_{i_lj_l},Y\}) =
  \PP_\sigma(\{X_{i_mj_m},Y\}).$$
  
\end{prop}

\begin{proof}

  Let $\mathcal K=\{x_{i_1j_1}, \dots, x_{i_kj_k}\}$ be a $k$-cherry on $\psi^*$, with MRCA the
  node $v$.  Then for any $y\in\mathcal X^*\smallsetminus\{x_{i_lj_l,}x_{i_mj_m}\}$,
  $\PP_\sigma(\{X_{i_lj_l},Y\}) =
  \PP_\sigma(\{X_{i_mj_m},Y\})$  by the exchangeability of
  $X_{i_lj_l}$ and $X_{i_mj_m}$.

  To see $\mathcal K$ is maximal with respect to this
  property, suppose $z\in \mathcal X^*\smallsetminus \mathcal K$.  (If no
  such $z$ exists, maximality is clear.) 
  We show $\mathcal K$ cannot be augmented by $z$ by showing that
  $\PP_\sigma(\{X_{i_1j_1},X_{i_2j_2}\}) \ne \PP_\sigma(\{Z,X_{i_2j_2}\})$ for some choice of
  $\lambda$. This then implies the same statement for generic values
  of $\lambda$, since these probabilities are polynomials in the
  exponentials of negative branch lengths.

  Choose all
  internal branch lengths of the species tree to be (near) 0 except
  for the branch $e$ above $v$, which we choose to have length (near)
  $\infty$. Consider the event $E$ that the $X_{i_1j_1}$ and $X_{i_2j_2}$ lineages
  coalesce on $e$ and are the first of the $\mathcal K$ lineages to do so.
  Then one sees that $$\PP_\sigma(\{X_{i_1j_i},X_{i_2j_2} \} )>\PP_\sigma(E)\approx
  \binom{k}{2}^{-1},$$ where the approximation becomes increasingly
  accurate as more extreme branch lengths are chosen.  However such
  choices of branch lengths make $\PP_\sigma(\{Z,X_{i_2j_2} \})$
 as close to 0 as
  desired, since the probability of the clade
  $\mathcal K$ goes to 1, and this is incompatible with clade $\{Z,X_{i_2j_2}\}$.
 Thus $\PP_\sigma(\{X_{i_1j_1},X_{i_2j_2}\}) \ne \PP_\sigma(\{Z,X_{i_2j_2}\})$.
  \smallskip

  To establish the converse,
  suppose now that $\mathcal K$ is maximal with
  respect to the stated property, but is not a $k$-cherry. By the above argument, maximality
  implies $\mathcal K$ is not a subset of any $l$-cherry for $l>k$.

  To achieve a contradiction, it is sufficient to show that there
  exist $x_{i_1j_1}, x_{i_2j_2} \in \mathcal K$, $y\in \mathcal X^*$ such that
  $\PP_\sigma(\{X_{i_1j_1},Y\}) \ne \PP_\sigma(\{X_{i_2j_2},Y\})$ unless branch lengths lie on
  a set of measure 0.  Let $v = \MRCA(\mathcal K)$.  Since $\mathcal
K$ is not contained in an $l$-cherry, there exists a non-leaf node
  $w$ which is a child of $v$. Moreover, $v$ is binary, since $\psi$
  is.

  Choose $x_{i_1j_1}\in  \mathcal K \smallsetminus\desc(w)$, which is non-empty
  because $v = \MRCA(\mathcal K)$. The node $w$ has at least two
  distinct leaf descendants, and since $v$ is binary at least one leaf
  descendant of $w$ must be in $\mathcal K$. Choose $x_{i_2j_2} \in \mathcal
 K \cap \desc(w)$, and $y\in \desc(w)\smallsetminus\{x_{i_2j_2} \}$.

  Let $m = | \desc (w)|$.  By choosing the length of the
  edge connecting $w$ and $v$ to be near $\infty$, and the length of
  all edges descended from $w$ to be near 0, the probability of clade
  $\{X_{i_1j_1},Y\}$ will be arbitrarily close to 0, and the probability of the
  clade $\{X_{i_2j_2},Y\}$ will be bounded below by a number arbitrarily close
  to $\binom{m}{2}^{-1}$, as in the argument above.  Thus, there exist
  branch lengths with $\PP_\sigma(\{X_{i_1j_1},Y\}) \ne \PP_\sigma(\{X_{i_2j_2},Y\})$. Because
  these probabilities are polynomials (in transformed branch lengths),
  the set of branch lengths where $\PP_\sigma(\{X_{i_1j_1},Y\}) = \PP_\sigma(\{X_{i_2j_2},Y\})$
  has measure 0.
 \end{proof}

Finally, we apply Proposition \ref{prop:kcherry} to show generic
 identifiability of species trees from clade probabilities when there
 are $\delta_i \ge 1$ lineages sampled for taxon $x_i$. (See
 Fig.~\ref{F:intra}.).

\begin{cor} 
  Let $\sigma^*=(\psi^*,\lambda,\delta)$ be an extended species tree,
  for which the pruned species tree $\psi$ is binary.  For generic
  choices of edge lengths $\lambda$, the topology of $\psi^*$ can be
  identified from the probabilities of clades under the
  multispecies coalescent.
\end{cor}

\begin{proof}
  All $k$-cherries of $ \psi^*$ can be identified by Proposition
  \ref{prop:kcherry} (although this is unnecessary if one assumes
  species assignments are given).  By assumption there are no other
  polytomies on $ \psi^*$; for any other clade $\mathcal A$ on the
  extended species tree, $v = \MRCA(\mathcal A)$ is a binary node. Thus Theorem
  \ref{thm:cladeinv} and Lemma \ref{lem:nonzero} apply.  These imply
  that for generic $\lambda$ all clades on the extended species tree
  can be identified, and hence $ \psi^*$ can be identified.
\end{proof}

Since this corollary does not assume that the assignment of individuals to
taxa is known in advance, it implies that under some circumstances
species assignment
can be deduced from clade probabilities. In particular, any
taxon for which three or more individuals are sampled will be identifiable
from $ \psi^*$. However taxa in which two individuals
  have been sampled will be indistinguishable from two taxa forming a
  2-clade with one individual sampled from each.

\section{Discussion}

We have shown that for generic branch lengths on a binary species
tree, it is possible to identify clades of the species tree, and
therefore the species tree topology, from probabilities of clades on gene trees.
More generally, we showed identifiability of clades consisting of taxa
descended from binary nodes even if the species tree is not itself
binary. In addition, we investigated how probable a clade on a gene tree must be
to infer it is also a clade on the species tree.

We have not shown the identifiability of branch lengths from clade
probabilities.  However, for any given species tree topology it is possible to write
systems of equations of clade probabilities as functions of the branch
lengths.  (As examples, consider the systems of equations for clade probabilities for some 4-taxon species trees shown in Table \ref{T:marginalize}.)
These systems are non-linear but polynomial in the transformed branch
lengths.  Since the number of branch length parameters is $n-2$ for an
$n$-taxon tree and there are $2^n-n-2$ non-trivial clade probabilities,
it is reasonable to expect such systems to be solvable, in principle,
for any sized tree. Although for particular small trees these can be 
solved, we have not found a general method applicable to arbitrary trees.
It thus remains conjectural that species tree branch lengths are identifiable from clade probabilities.
If multiple individuals are sampled from some species, then the species tree has additional branch length
parameters (for branches leading to such species), but will have an even larger numbers
of clades probabilities that could conjecturally be used to estimate branch lengths.

While the invariants of Theorem \ref{thm:cladeinv} are useful in proving
identifiability of a species tree topology, they do not immediately indicate 
a practical way to infer the species tree from clade
probabilities. In particular, each term in the invariant of Theorem
\ref{thm:cladeinv} is the probability 
that a random gene tree has a clade that is not a clade on the species tree.
The clade probabilities needed for the invariant of Theorem \ref{thm:cladeinv}
may therefore be quite small. 
For species trees with moderately long branches, many of
these probabilities could be difficult to estimate from
finite data sets. However, in such a situation the results of Section
\ref{sec:high} might offer an alternative way of inferring species
tree clades as those which occur with high frequency on gene
trees. This suggests the possibility of a hybrid
approach in which one accepts highly probable clades as being
clades on the species tree, as in a greedy algorithm, yet exploits the 
symmetries of clade probabilities expressed by invariants to determine 
other species clades. Thus our identifiability results should motivate further research on
species tree inference methods that are statistically consistent and that can
outperform greedy consensus on typical data sets with imperfectly
estimated clade probabilities.

\section*{Acknowledgement}
\label{sec:acknowledgement}

The authors thank the Statistical and Applied Mathematical Sciences
Institute, where this work was begun during its 2008-09 program on
Algebraic Methods in Systems Biology and Statistics, and the Institut Mittag-Leffler,
where the writing was completed during the Spring 2011 program on Algebraic Geometry with a View Towards Applications.
ESA and JAR were
supported by funds from the National Science Foundation, grant DMS
0714830, and JAR by an Erskine Fellowship from the University of
Canterbury. JHD was funded by the New Zealand Marsden Fund. All
authors contributed equally to this work.
 
{\footnotesize 
\bibliographystyle{plain}
\bibliography{bibfile}
}

\appendix

\section{Clade probabilities for subtrees as linear combinations}\label{ap:linear}

An essential difficulty in dealing with clade probabilities in
mathematical arguments is that it is not easy to see relationships
between probabilities of clades on gene trees arising from a species
tree on a set of taxa and the clade probabilities on the induced
gene trees obtained by restricting the set of taxa to a smaller set. 
This frustrates the common approach used to prove results for
large trees, by inductive arguments on the number of taxa.

Consider a set of taxa $\mathcal X$, a proper subset $\mathcal Y
\subset \mathcal X$, and a species tree $\sigma = (\psi,\lambda)$ on
$\mathcal X$.  We show here that, in general, probabilities of gene
tree clades for the induced species tree on $\mathcal Y$,
$\sigma|\mathcal Y$, cannot be written as the same linear combination
of gene tree clade probabilities for $\sigma$ for all choices of
$\psi$.  Thus there is no linear formula for the clade probabilities
for the smaller taxon set that does not depend on the species tree
topology.

This is in contrast to, for example, gene tree probabilities for
$\sigma| \mathcal Y$, which can be written as linear combinations of
gene tree probabilities for $\sigma$, where the weight assigned to
each gene tree probability in the combination has no dependency on
$\sigma$.

\medskip

To show that the same linear combination cannot be used to marginalize
clade probabilities independently
of the species tree, we consider three
species tree topologies on four taxa, as given in Table
\ref{T:marginalize} : $(((a,b),c),d)$, $((a,d),(b,c))$, and
$((a,b),(c,d))$.  There are 10 non-trivial clades, so any
linear combination of clade probabilities $c_1, \dots, c_{10}$ has the
form
\begin{equation}\label{E:lc}
\sum_{i=1}^{10} \alpha_i c_i
\end{equation}
for some $\alpha_1, \dots, \alpha_{10}$.

\begin{table}
\begin{footnotesize}
\begin{center}
\caption{Probabilities of clades under three 4-taxon species trees. $ X=\exp(-x), Y=\exp(-y)$.}\label{T:marginalize}
\begin{tabular}{c r r r}
\hline & \multicolumn{3}{c}{probability under species tree}\\
clade & $(((a,b)\tc x,c)\tc y,d)$ & $((a,d)\tc x,(b,c)\tc y)$ & $((a,b)\tc x,(c,d)\tc y)$\\
\hline $c_1 = \mathbb P(AB)$ &  $1 - \frac{2}{3}X - \frac{1}{9}XY^3$ & $\frac{2}{9}XY$ & $1 - \frac{2}{3}X - \frac{1}{9}XY$\\
$c_2 = \PP_\sigma(AC)$ & $\frac{1}{3}X - \frac{1}{9}XY^3$ & $\frac{2}{9}XY$ & $\frac{2}{9}XY$\\
 $c_3 = \PP_\sigma(AD)$ &$\frac{1}{6}XY + \frac{1}{18}XY^3$ &$1 - \frac{2}{3}X - \frac{1}{9}XY$ & $\frac{2}{9}XY$\\
 $c_4 = \PP_\sigma(BC)$ &$\frac{1}{3}X - \frac{1}{9}XY^3$ & $1 - \frac{2}{3}Y - \frac{1}{9}XY$ & $\frac{2}{9}XY$\\
 $c_5 = \PP_\sigma(BD)$ & $\frac{1}{6}XY + \frac{1}{18}XY^3$ & $\frac{2}{9}XY$ & $\frac{2}{9}XY$\\
 $c_6 = \PP_\sigma(CD)$ & $\frac{1}{3}Y - \frac{1}{6}XY +\frac{1}{18}XY^3$& $\frac{2}{9}XY$ & $1 - \frac{2}{3}Y - \frac{1}{9}XY$\\
 $c_7 = \PP_\sigma(ABC)$ &$1 - \frac{2}{3}Y - \frac{1}{3}XY + \frac{1}{6}XY^3$& $\frac{1}{3}X - \frac{1}{6}XY$ & $\frac{1}{3}Y - \frac{1}{6}XY$\\
 $c_8 = \PP_\sigma(ABD)$ & $\frac{1}{3}Y - \frac{1}{6}XY$ &$\frac{1}{3}Y - \frac{1}{6}XY$ &  $\frac{1}{3}Y - \frac{1}{6}XY$\\
  $c_9 = \PP_\sigma(ACD)$ & $\frac{1}{6}XY$ &$\frac{1}{3}Y - \frac{1}{6}XY$ & $\frac{1}{3}X - \frac{1}{6}XY$\\
 $c_{10} = \PP_\sigma(BCD)$ & $\frac{1}{6}XY$ &$\frac{1}{3}X - \frac{1}{6}XY$ &  $\frac{1}{3}X - \frac{1}{6}XY$\\
 \hline
 \end{tabular}
\end{center}
\end{footnotesize}

\end{table}

We consider obtaining the probability of clade $CD$ when the taxon set is restricted to $\{a,c,d\}$ (\emph{i.e.}, marginalizing over taxon $b$). 
Assuming first that the species tree is $((a,d)\tc x,(b,c)\tc y)$,  the restricted species tree is $((a,d)\tc x,c)$, and the probability of clade $CD$ is $\frac{1}{3}X$.  If a linear combination of the clade probabilities on the larger tree is to yield this probability, then by inserting the formulas for the $c_i$ from Table \ref{T:marginalize} into 
Eq.~\eqref{E:lc}
and equating coefficients, we obtain the following equations:

\begin{align*}
\alpha_3 + \alpha_4 &= 0,\\
-2\alpha_3 + \alpha_7 + \alpha_{10} &= 1,\\
-2\alpha_4 + \alpha_8 + \alpha_9 &= 0,\\
4\alpha_1 + 4\alpha_2 -2 \alpha_3 - 2\alpha_4 +4 \alpha_5 + 4\alpha_6 - 3\alpha_7 -3\alpha_8 -3\alpha_9 -3 \alpha_{10} &= 0,
\end{align*}
where the rows correspond to the coefficients of $1$, $X$, $Y$, and $ XY$.  The system is underdetermined since there are 10 unknowns and only four equations.  

Similar systems can be obtained by considering 
other species trees.  For the other species trees in Table \ref{T:marginalize}, $(((a,b)\tc x,c)\tc y,d)$ and $((a,b)\tc x,(c,d)\tc y)$, respectively, restricting to taxa $\{a,c,d\}$ leads to 
trees $((a,c)\tc y,d)$ and $(a,(c,d)\tc y)$, and
probabilities of  clade $CD$  that are $\frac{1}{3}Y$  and $1-\frac{2}{3}Y$. Equating coefficients on all three species trees  in Table \ref{T:marginalize}, we have the equations encoded by the following $13\times 11$ augmented matrix:
\begin{scriptsize}
$$ \left[ \begin {array}{rrrrrrrrrrr} 1&0&0&0&\phantom{-}0&0&1&0&0&0&0
\\\noalign{\medskip}-2&1&0&1&0&0&0&0&0&0&0\\
\noalign{\medskip}0&0&0&0&0&1&-2&1&0&0&1
\\\noalign{\medskip}0&0&1&0&1&-1&-2&-1&1&1&0
\\\noalign{\medskip}-2&-2&1&-2&1&1&3&0&0&0&0
\\\noalign{\medskip}0&0&1&1&0&0&0&0&0&0&0
\\\noalign{\medskip}0&0&-2&0&0&0&1&0&0&1&1
\\\noalign{\medskip}0&0&0&-2&0&0&0&1&1&0&0
\\\noalign{\medskip}4&4&-2&-2&4&4&-3&-3&-3&-3&0
\\\noalign{\medskip}1&0&0&0&0&1&0&0&0&0&1
\\\noalign{\medskip}-2&0&0&0&0&0&0&0&1&1&0
\\\noalign{\medskip}0&0&0&0&0&-2&1&1&0&0&-2
\\\noalign{\medskip}-2&4&4&4&4&-2&-3&-3&-3&-3&0
\end {array} \right].$$
\end{scriptsize}
\noindent Here rows 1--5 represent the system of equations implied by the species tree $(((a,b),c),d)$, 
rows 6--9 represent the system of equations corresponding to the species tree $((a,d),(b,c))$,
and rows 10--13 represent the system of equations corresponding to the species tree $((a,b),(c,d))$.    Gaussian elimination shows this
system of 13 equations  is inconsistent.


\section{Additional clade invariants for small trees}\label{app:moreinv}

For trees on five or fewer taxa, computations of a Gr\"obner basis for
invariants in clade probabilities show that the construction
of Theorem \ref{thm:cladeinv} fails to produce all invariants, or even
all linear ones.   In this appendix, we indicate the results of such computations 
that we performed using the software \texttt{Singular} \cite{GPS09}.
We emphasize that by linear invariant we mean linear homogeneous
invariant, so that the trivial invariant, which is
inhomogeneous, is not counted when we give dimensions of spaces.

\smallskip

For the 3-taxon tree there is only a single invariant, the linear one arising from
cherry-swapping, produced by Theorem \ref{thm:cladeinv}.

\smallskip

For the 4-taxon balanced tree topology $((a,b),(c,d))$, there is a
6-dimensional space of linear invariants, yet the ones constructed in
Theorem \ref{thm:cladeinv} span only a 5-dimensional subspace. The
additional generator needed to obtain all linear invariants can be
taken to be
$$\PP_\sigma(AB)-\PP_\sigma(CD)-2\PP_\sigma(ABC)+2\PP_\sigma(ACD).$$
The ideal of all invariants has just one additional generator, which is quadratic.

\smallskip

For the $4$-taxon caterpillar tree topology $(((a,b),c),d)$, there is a
$5$-dimensional space of linear invariants. However the construction of Theorem
\ref{thm:cladeinv} produces only a $4$-dimensional space of linear
invariants. For the full space of linear invariants, the polynomial
$$\PP_\sigma(AB)+2\PP_\sigma(AC)+9\PP_\sigma(CD)-\PP_\sigma(ABC)-11\PP_\sigma(ABD)-4\PP_\sigma(ACD).$$
can be taken as the missing generator.

In addition, there were one quadratic and three cubic polynomials in a full Gr\"obner basis.

\smallskip

For the 5-taxon balanced tree topology $(((a,b),c),(d,e))$, the construction of
Theorem \ref{thm:cladeinv} produces a 14-dimensional subspace within a
16-dimensional space of linear invariants. Additional generators can
be taken to be
\begin{multline*}22\PP_\sigma(CD)+5\PP_\sigma(DE )-5\PP_\sigma(ABC)-22\PP_\sigma(ABD)
+15\PP_\sigma(CDE)\\
+10\PP_\sigma(ABCD)-25\PP_\sigma(ABDE)
-20\PP_\sigma(ACDE)
\end{multline*}
and
\begin{multline*}
11\PP_\sigma(AB)+22\PP_\sigma(AC )2-25\PP_\sigma(DE )+14\PP_\sigma(ABC )
-22\PP_\sigma(ABD )
-44\PP_\sigma(ACD )\\+24\PP_\sigma(CDE )-50\PP_\sigma( ABCD)
+4\PP_\sigma(ABDE )+56\PP_\sigma(ACDE ).
\end{multline*}

In addition to the linear invariants,
there are eight quadratic invariants and $13$ cubic invariants in a Gr\"obner
basis for the ideal.

\smallskip

For the 5-taxon psuedo-caterpillar tree topology $(((a,b),(d,e)),c)$, the
construction of Theorem \ref{thm:cladeinv} produces a 13-dimensional
subspace within a 14-dimensional space of linear invariants. An
additional generator can be taken to be
$$\PP_{\sigma}(AB)-\PP_{\sigma}(DE)-6\PP_{\sigma}(ABC)-2\PP_{\sigma}(ABD)+2\PP_{\sigma}(ADE)+6\PP_{\sigma}(CDE).$$
The algorithm for computing the full ideal of invariants for this topology did not
terminate in a reasonable amount of time, so the full ideal remains unknown.   Partial
computations in which the degree of generators is bounded show that there are generators
in degrees $2$, $3$, $4$, $5$, and $6$, in addition to linear invariants.

\smallskip

For the 5-taxon caterpillar tree $((((a,b),c),d),e)$, the construction
above produces a 11-dimensional subspace within a 12-dimensional space
of linear invariants. One choice for the additional generator is
\begin{multline*}5\PP_\sigma( AB)+10\PP_\sigma(AC )+24\PP_\sigma(CD )+62\PP_\sigma( DE)+2\PP_\sigma( ABC)
-20\PP_\sigma(ABD )\\
-29\PP_\sigma(ABE )+8\PP_\sigma(ACD )-58\PP_\sigma(ACE )
+45\PP_\sigma(CDE )-7\PP_\sigma(ABCD )\\-76\PP_\sigma(ABCE )-44\PP_\sigma(ABDE )
+2\PP_\sigma(ACDE ).
\end{multline*}
Our attempt to compute a Gr\"obner basis for the caterpillar
topology did not terminate in a reasonable amount of time.  We did, however, find 
quadratic generators in addition to the linear ones, but found no higher
degree generators.  It is reasonable to speculate that the full ideal is
generated in degree one and two for this topology.

\smallskip

It would be quite interesting 
to find general constructions that lead to the additional linear 
invariants not explained by Theorem \ref{thm:cladeinv}. Similarly, understanding
the structure of higher degree invariants by non-computational means
is an open challenge.

\end{document}